\documentclass{article}

\usepackage[utf8]{inputenc}

\usepackage{amsmath}
\usepackage{amsfonts}
\usepackage{amssymb}
\usepackage{amsthm}
\usepackage{mathabx}
\usepackage{geometry}
\usepackage{graphicx}
\newcommand{\hlin}{\operatorname{H}_{\operatorname{lin}}}

\newcommand{\bin}[3]{\operatorname{\mathbf{bin}}({#1}, {#2}, {#3})}
\newcommand{\lbin}[2]{\operatorname{\mathbf{lbin}}({#1}, {#2})}
\newcommand{\vbin}[2]{\operatorname{\mathbf{bin}}({#1}, {#2})}
\newcommand{\vlbin}[1]{\operatorname{\mathbf{lbin}}({#1})}

\newcommand{\probs}[2]{\operatorname{\mathbf{Pr}}_{{#1}}\left[{#2}\right]}
\newcommand{\prob}[1]{\probs{}{#1}}
\newcommand{\expects}[2]{\operatorname{\mathbf{E}}_{{#1}}\left[{#2}\right]}
\newcommand{\expect}[1]{\expects{}{#1}}

\newtheorem{lemma}{Lemma}
\newtheorem{theorem}{Theorem}

\newtheorem{corollary}{Corollary}

\title{Expected number of uniformly distributed balls in a most loaded bin using placement with simple linear functions}
\author{Martin Babka\thanks{Research supported by the Czech Science Foundation grant GA14-10799S.}}

\begin{document}

\maketitle

\begin{abstract}
We estimate the size of a most loaded bin in the setting when the balls are placed into the bins using a random linear function in a finite field.
The balls are chosen from a transformed interval. 
We show that in this setting the expected load of the most loaded bins is constant.

This is an interesting fact because using fully random hash functions with the same class of input sets leads to an expectation of $\Theta\left(\frac{\log m}{\log \log m}\right)$ balls in most loaded bins where $m$ is the number of balls and bins.

Although the family of the functions is quite common the size of largest bins was not known even in this simple case.
\end{abstract}

\section{Introduction}

Our basic task is to estimate estimate the size of a largest bin in a special case of the balls and bins model. This models simply means that the balls are randomly thrown into bins. The process of their placement is of a various study -- its randomness, independence and other properties lead to various bin sizes.
The most simple model is to use fully random functions or some kind of their approximation to place the balls. There is a plenty of results, i.e. estimates of bin sizes, for various placement processes.

When the balls are thrown independently at random to the bins the expected size of the largest bin is $\Theta\left(\frac{\log m}{\log \log m}\right)$.

One of the first results were shown by Carter and Wegman \cite{cw} and this model was used to design universal and perfect hashing. They showed that the expected size of a bin is a constant when the placement is done by the functions which we will refer to as simple linear functions. These functions are two-wise independent and thus achieve $O(\sqrt{m})$ expected size of a largest bin.

It is also possible to use functions with higher degrees of independence and obtain better bounds. There are lower bounds for on the speed of such functions, size needed to represent and the size of the largest bin and independence they achieve given by Siegel \cite{siegel}.

The need to improve the size of the largest bins lead to two-choice paradigm. Out of two bins, hence we use two functions, the balls is placed into the smaller one. In this model the size of the largest bin is $O(\log \log m)$ where $m$ is the number of balls and bins shown by Azar et al \cite{azar} and improved by V\"{o}cking \cite{vocking}.

Nowadays more complicated family of functions are studied in \cite{wieder}. The functions no longer rely on high degree of independence but are designed so that they achieve small largest bins even with high probability.

Our model exhibits the use of simple linear functions and the balls are chosen from an interval in $\mathbb{Z}_p$. Such model has a constant size of largest bins.

\section{Notation and definitions}
\label{sec:notation}
We refer to the set $\{0, \dots, k - 1\}$ as to $[k]$. 
In the whole text we assume that $p$ is a fixed prime. 
The set of chosen balls is denoted by $S \subset [p]$.
The number of bins is the same as the number of balls and is denoted by $m$, i.e. $|S| = m$.

For each pair $(a, b) \in [p]^2$ we define the function $h'_{a, b}$ as $h'_{a, b}(x) = (ax + b) \bmod p$ and the function $h_{a, b}$ as $h_{a, b}(x) = h'_{a, b}(x) \bmod m$.

The multiset of simple linear functions mapping $[p]$ to the range $[m]$ is denoted by $\hlin$ and is defined as $\hlin = \{h_{a, b} \mid a, b \in [p] \}$.
For a function $h \in \hlin$ we define the size of $i$-th bin as $\bin{h}{S}{i} = |S \cap h^{-1}(i)|$ and the maximal size of the bin as $\lbin{h}{S} = \max_{i \in [m]} \bin{h}{S}{i}$.

In the following text we fix the probability space to be formed by a uniform choice of $h \in \hlin$.
The notation $\vbin{S}{i}$ and $\vlbin{S}$ then refers to the random variables formed by the mentioned random uniform choice.

For an element $x$ we define the value $l(x, a, b) = \left\lfloor\frac{ax + b}{p}\right\rfloor$; that is how many ``leaps'' are created by applying the function $h_{a, b}$ on the element $x$ in the field $\mathbb{Z}_p$.

\section{Collision probability for three elements}

We first study the probability of collision of three arbitrary elements.
By collision of the elements we understand the event when all of the elements are mapped to the same element in $[m]$ by the randomly chosen linear function.

We fix three different elements $x, y, z \in [p]$ and we count the number of pairs $(a, b) \in [p]^2$ such that $|h_{a, b}(\{x, y, z\})| = 1$.

We start by simplifying to the case when $x = 0, y = 1$ and the third element $z = d$ for a suitable $d \in [p]$ such that $d > 1$ depending on the choice of $x, y, z$.

\begin{lemma}[Transformation lemma]
\label{lemma:transformation}
Let $x, y, z \in [p]$ be arbitrary different elements. Moreover assume that $i_x, i_y, i_z \in [m]$. Then there exist an element $d \in [p]$ such that
\[
\prob{h(x) = i_x, h(y) = i_y,  h(z) = i_z} = \prob{h(0) = i_x, h(1) = i_y, h(d) = i_z}.
\]
\end{lemma}
\begin{proof}
The idea of the proof is simple. We show that there is a one-to-one map between simple linear functions mapping $x, y, z$ to $i_x, i_y, i_z$ and simple linear functions transforming $0, 1, d$ to the same elements.

In the first part of the proof we observe that combining simple linear functions with a linear function in $\mathbb{Z}_p$ does not change the the probability space.
There is a single linear function transforming $0, 1$ to $x, y$ in $\mathcal{Z}_p$ which we refer to as $h'_{\alpha,\beta}$.
Finally we choose $d$ so that $h'_{\alpha, \beta}(d) = z$ and the proof is finished.

We show that the elements $x, y, z$ can be transformed to the elements $0, 1, d$ so that the probability of the mappings from the statement of the lemma remains the same.

Choose $\alpha, \beta \in [p]$ so that $\alpha \neq 0$.
Observe that the mapping $(\gamma, \delta) \mapsto (\alpha \gamma, \beta \gamma + \delta)$ is a one-to-one map on $[p]^2$.
If there is another pair $(\epsilon, \phi)$ such that $(\alpha \epsilon, \beta \epsilon + \phi) = (\alpha \gamma, \beta \gamma + \delta)$, then $\gamma = \epsilon$ and $\delta = \phi$. Thus the mapping is injective.
Also for arbitrary $(r, s) \in [p]^2$ the element $(\alpha^{-1}r, s - \beta\alpha^{-1}r)$ is mapped to $(\alpha \alpha^{-1}r, \beta r + s - \beta\alpha^{-1}r) = (r, s)$.

The compound function $h'_{a, b} \circ h'_{\alpha, \beta}$ is exactly equal to the function $h'_{\alpha a, \beta a + b}$; this also follows from the fact that the set of all linear functions in $\mathbb{Z}_p$ forms a group with the operation of compounding functions. 

Let $H' = \{h'_{a, b} \mid (a, b) \in [p]^2\}$.
From the previous we can conclude that the combination of a function $h' \in H'$ with a fixed function $h'_{\alpha, \beta}$ is a one-to-one map in the space of functions $H'$.
Also observe that the composition of a function $h_{a, b} \in \hlin$ with $h'_{\alpha, \beta}$ can not change the probability (count of the functions) of mapping arbitrary three elements to a their prescribed images.

There is also a single function $(\alpha, \beta) \in [p] ^ 2$, i.e. a single function $h'_{\alpha, \beta}$, transforming the elements $0$ and $1$ to $x$ and $y$ in the field $[p]$ without taking modulo $m$. It is the function $\beta = x$ and $\alpha = y - x$.
To prove the lemma we choose $d \in [p]$ such that $h'_{\alpha, \beta}(d) = z$, i.e. $d = \alpha ^ {-1}(z - \beta)$.
\end{proof}

Lemma \ref{lemma:transformation} shows that the probability properties, e.g. collision, mapping to the prescribed elements, for the elements $x, y, z$ are the same as for the elements $\{0, 1, d\}$ where $d$ comes from the previous lemma.

Next we estimate the collision probability for the elements $0, 1, d$.

\begin{lemma}[Probability of collision of three elements]
\label{lemma:probability-3-elements}
Let $d \in [p]$ be arbitrary element.
\[
\prob{|h(\{0, 1, d\})| = 1 } = \frac{1 + \max\left(1, \frac{p}{dm}\right)\left(1 + \frac{d}{m}\right)}{p}.
\]
\end{lemma}
\begin{proof}
We count the number of functions $h \in \hlin$ such that $h(\{0, 1, d\}) = \{y\}$ for some $y \in [m]$.
For each $x \in [p]$ it holds that $l(x, a, b) \in [x]$ and $h(x) = (ax + b - l(x, a, b)p) \bmod m$.

Whenever the elements $0, 1$ and $d$ are mapped to the same element $y$ it must hold that $h(0) = h(1)$ and $h(0) = h(d)$. Hence
\begin{align*}
	b \bmod p \bmod m & = (a + b) \bmod p \bmod m \\
	b \bmod p \bmod m & = (da + b) \bmod p \bmod m. \\
\end{align*}
From which we obtain the following sequence of equations
\begin{align*}
	m & \mid a + b - l(1, a, b)p - b \\
	m & \mid da + b - l(d, a, b)p - b, \\
\end{align*}
\begin{align*}
	m & \mid a - l(1, a, b)p \\
	m & \mid da - l(d, a, b)p, \\
\end{align*}
\begin{align*}
	m & \mid da - dl(1, a, b)p \\
	m & \mid (dl(1, a, b) - l(d, a, b))p. \\
\end{align*}
Since $p$ is a prime we conclude the fact that $m \mid dl(1, a, b) - l(d, a, b)$.
We estimate the collision probabilities from the two statements following from the previous formulas:
\begin{align}
	m & \mid a - l(1, a, b)p \label{3-prob-1-statement} \\
	m & \mid (dl(1, a, b) - l(d, a, b))p. \label{3-prob-2-statement}
\end{align}

The statement (\ref{3-prob-2-statement}) roughly means that out of $d$ possible values for $l(d, a, b)$ only the $1 / m$ fraction may generate the collision of the three elements. Notice that for a fixed $l \in [d]$ it holds that $\{a \in [p] \mid l(d, a, b) = l\}$ equals is a subinterval of $[p]$.
From (\ref{3-prob-1-statement}) we can observe that only the $1 / m$ fraction from the possible values of $a$ lying in the appropriate intervals allowed by valid values of $l(d, a, b)$ are causing collisions.

For the rest of the proof fix the value of $b$. 
First, we show that the values of $a$ such that $l(d, a, b) = l \in [d]$ form disjoint intervals in $[p]$ each of size at most $\lceil p/d \rceil$.
Then we count the number of values $a$ in an interval causing collisions -- using (\ref{3-prob-1-statement}).
And finally we count the number of the valid intervals.

Let $l(d, a, b) = l$, then it holds that $l \leq \frac{da + b}{p} < l + 1$. Immediately we get that $a \in \left[\frac{pl - b}{d}, \frac{p(l + 1) - b}{d}\right) \cap \mathcal{Z}$. The total number of values of $a$, i.e. integers, in each valid interval is at most $\lceil p / d \rceil$. The ceiling must be applied. For example assume an interval of length of 1.5 starting at point 0.8 -- it contains two integer points 1 and 2. This happens whenever $\frac{pl - b}{d}$ is an integer.

Now fix the value $l \in [d]$ such that $l(d, a, b) = l$.
In order to estimate the number of values of $a$ causing the collisions we split into two cases according to the value of $l(1, a, b)$.

\subparagraph{The first case, $l(1, a, b) = 0$.} 
From the two previous statements we conclude that
\begin{align*}
	m & \mid a \\
	m & \mid l. \\
\end{align*}

\subparagraph{The second case, $l(1, a, b) = 1$.}
As in the first case it must hold that
\begin{align*}
	m & \mid a - p \\
	m & \mid d - l. \\
\end{align*}

In both cases, there are at most $\lceil d/m \rceil$ values of $l$ satisfying the second condition.
Also for each satisfying value of $l$ there are at most $\lceil \lceil\frac{p}{d}\rceil / m \rceil$ values of $a$ causing the collision.

In both cases and for each $b$ it holds that the probability of collision of the three elements is bounded by
\[
\frac{\left(1 + \frac{1 + \frac{p}{d}}{m}\right)\left(1 + \frac{d}{m}\right)}{p} =
\begin{cases}
	\frac{1}{m^2} + \frac{1}{dm} + O\left(p^{-1}\right) & \mbox{if } p/dm \geq 1 \\
	\frac{1}{m^2} + \frac{d}{pm} + O\left(p^{-1}\right) & \mbox{otherwise, i.e. } \frac{p}{m} \leq d < p.
\end{cases}
\]
\end{proof}

The worst possible case is for $d = 2$ and the probability is roughly $1/2m$. 
When $d \geq p/m$, the formula is a great overestimate as shown in Figure~1.

\begin{figure}[h]
	\label{fig:probability-3}
	\centering
	\includegraphics[width=8cm]{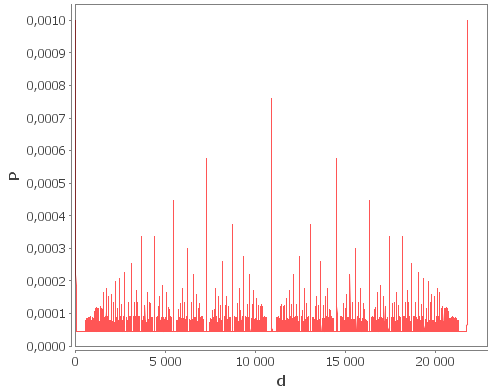}
	\caption{The function of probability of collision of the elements $0, 1, d$ with respect to $d$. Notice that the probability is decreasing in the part when $d \leq p / m$ and is almost symmetric. In this figure $m = 512$ and $p = 21787$.}
\end{figure}

\begin{corollary}
\label{co:d-elements}
Let $d < p / m$. Then $\prob{|h([d])| = 1} \leq \frac{1}{(d - 1) m} + 1/m^2 + O(p^{-1})$.
\end{corollary}
\begin{proof}
When all the elements from $[d]$ collide, then the elements $\{0, 1, d\}$ must collide as well. The probability of the collision of $\{0, 1, d\}$ is hence a valid upper bound on the probability of the collision of the whole interval. The statement is then a direct application of Lemma~\ref{lemma:probability-3-elements}.
\end{proof}

For completeness we just show a simple fact that our probability estimate is tight when we have a stronger assumption, namely we assume $p > 3m^2$.

\begin{lemma}
\label{lm:0-d-prob-lower-bound}
If $d \leq m$ and $p > 3m^2$, then $\prob{|h(\{0, 1, \ldots, d - 1\})| = 1} = \Omega\left(\frac{1}{dm}\right).$
\end{lemma}
\begin{proof}
For a fixed $b$, if $a < (p - b)/d$ and $m \mid a$, then the elements $\{0, 1, \ldots, d - 1\}$ collide.
For each $b$ there are at least $\lfloor (p - b)/dm \rfloor$ such values of $a$.

We conclude that the number of pairs $(a, b)$ making the elements collide is at least
\[
\sum_{b \in [p]} \frac{p - b}{dm} - 1 = \frac{(p + 1)p}{2dm} - p \geq \frac{p ^ 2 - 2pdm}{2dm} \geq \frac{p^2}{6dm}.
\]

Thus the resulting probability is at least $\frac{1}{6dm}$.
\end{proof}

\section{The expected size of most loaded bins}

First we study the role of the parameter $b$ in the hash function $h_{a, b}$.

The following lemma states that the effect of $b$ on $\vlbin{S}$ is not asymptotic since it more or less only shifts the largest bin.
\begin{lemma}
\label{lm:b-zero}
Assume that  $a, b \in [p]$ and $S \subseteq [p]$. Then \[ \frac{1}{2} \lbin{h_{a, b}}{S} \leq \lbin{h_{a, 0}}{S} \leq 2\lbin{h_{a, b}}{S} . \]
\end{lemma}
\begin{proof}
Let $L \subseteq S$ be elements of bin $y$, i.e. $h_{a, b}(L) = y$.
For each $x \in L$ we have that 
\begin{align*}
h_{a, 0}(x) 
	& = ax \bmod p \bmod m \\ 
	& = (ax + b - b) \bmod p \bmod m \\ 
	& = \begin{cases}
((ax + b) \bmod p - b \bmod p)\bmod m & \mbox{ if } (ax + b) \bmod p \geq b \\
(p + (ax + b) \bmod p - b \bmod p) \bmod m & \mbox{ otherwise.} \\
\end{cases}
\end{align*}

Notice that the two possible new bins are either $(y - b) \bmod m$ or $(p + y - b) \bmod m$.
The lemma now follows from the following two observation.
First each original bin is either shifted and keeps its size or is split into two possibly uneven shifted bins -- hence $\frac{1}{2} \lbin{h_{a, b}}{S} \leq \lbin{h_{a, 0}}{S}$.
And notice that each new bin can only contain elements from at most two different original bins and thus $\lbin{h_{a, 0}}{S} \leq 2\lbin{h_{a, b}}{S}$.
\end{proof}

For completeness let us mention that the change of the sign of $a$ has almost no effect on $\vlbin{S}$.
\begin{lemma}
\label{lemma:sign-a}
Assume that  $a \in [p]$ and $S \subseteq [p]$ such that $0 \not\in S$. Then \[ \lbin{h_{a, 0}}{S} = \lbin{h_{p - a, 0}}{S} . \]
\end{lemma}
\begin{proof}
Similarly as in the proof of the previous lemma. Let $L \subseteq S$ be elements of bin $y$, i.e. $h_{a, 0}(L) = y$.
Let $x \in L$, then $h_{p - a, 0}(x) = (p - a)x \bmod p \bmod m = (p - ((ax) \bmod p))\bmod m = (p - y) \bmod m$.
Observe that $(p - a)x \bmod p = p - ((ax) \bmod p)$ holds only when $x \neq 0$.
The bin $y$ is thus moved to the bin $(p - y) \bmod m$ and the lemma holds.
\end{proof}

Obviously allowing zero makes only a negligible change.
\begin{corollary}
Let $S \subseteq [p]$, then
\[
\lbin{h_{a, 0}}{S} -1 \leq \lbin{h_{p - a, 0}}{S} \leq \lbin{h_{a, 0}}{S} + 1.
\]
\end{corollary}

For the choice of $S = [m]$ we show that the expected size of a most loaded bin is within $O(1)$. 
This can be compactly formulated as follows.
\begin{theorem}
\label{thm:interval-constant}
Assume that $p \geq m^2$, then
\[
\expect{\vlbin{[m]}} = O(1).
\]
\end{theorem}
\begin{proof}
By Lemma~\ref{lm:b-zero} we may assume that the chosen function has $b = 0$ without asymptotically increasing the expected size of the largest bin. In the proof of the claims we thus assume that the chosen linear function is exactly the function $h_{a, 0}$. Moreover we assume that $a \neq 0$. Notice that this assumption adds exactly $m/p$ to the computed expected value which is $O(1)$.

Observe that each bin is formed by a single arithmetic progression. Notice that since $p$ is a prime it holds that $(-p) \bmod m$ is co-prime with $m$.
The reason can be stated as follows.
Let $x_1 < x_2$ be two elements in a single bin, then for $d = x_2 - x_1$ it holds that $m \divides ad \bmod p$ or $m \divides p - (ad \bmod p)$.

All the solutions of the equation $ax \bmod p \bmod m = 0$ where $x \in [m]$ form a finite arithmetic progression.
For the proof of the previous statement notice that since $p$ is a prime it holds that $(-p) \bmod m$ is co-prime with $m$. 

In addition a difference $d$ and a given length $l$, $l \geq 3$, there is a canonical value $x \in [m]$ such that if there is a bin of size at least $l$, then there is another bin formed by an arithmetic progression of length at least $l$ with the same difference $d$ having $x$ as the minimal element. If $ad \bmod p < p/2$, we choose $x = \operatorname{argmin}_{x \in [m - ld]} ax \bmod p$. Otherwise we put $x = \operatorname{argmax}_{x \in [m - ld]} ax \bmod p$.

After establishing the previous facts we simply compute the expected value of $\vlbin{[m]}$ using the following idea. Now we allow $b$ to have arbitrary value.

Assume that $\vlbin{[m]} > l \geq 3$, then there is an arithmetic progression chosen from $[m]$ of size at least $l/2$ collapsing into a single bin, here we use Lemma~\ref{lm:b-zero}. Since for a fixed difference and length we have its canonical position there are at most $m/l$ possible arithmetic progressions from which we choose from. By Corollary~\ref{co:d-elements} we upper bound the probability of the collapse of the arithmetic progression as 
\[
\frac{m}{l} \left(\frac{1}{(l/2 - 1)m} + 1/m^2 + O(p^{-1})\right) \leq O(l^{-2}).
\]

Hence for $l \geq 3$ we have
\[
\prob{\vlbin{[m]} \geq l} = O(l^{-2}).
\]

Then we simply conclude that
\[
\expect{\vlbin{[m]}} \leq O(1) + \sum_{l = 1}^m O\left(\frac{1}{l^2}\right) = O(1).
\]

\end{proof}

We can conclude the main result, i.e. each set transformable to $[m]$ in $\mathbb{Z}_p$ has constant sized largest bins.
\begin{corollary}
Let $S \subseteq [p]$, $a, b \in [p]$. 
If $\forall x \in [m] \colon (ax + b) \bmod p \in S$, then $\expect{\vlbin{S}} = O(1)$.
\end{corollary}
\begin{proof}
Direct corollary of Theorem~\ref{thm:interval-constant} since by Lemma~\ref{lemma:transformation} (extended to all the elements of $S$) the probabilistic properties of $S$ do not change under the transformation $x \mapsto (ax + b) \bmod p$.
\end{proof}


\begin{thebibliography}{32}
\bibitem{cw}
J.L. Carter, and M.N. Wegman. 
\newblock Universal Classes of Hash Functions.
\newblock Journal of Computer and System Sciences, 18. pages 143--154, 1979.

\bibitem{siegel}
A. Siegel. 
\newblock On universal classes of extremely random constant-time hash functions.
\newblock SIAM Journal on Computing, 33(3). pages 505--543 (electronic), 2004

\bibitem{wieder}
L.E. Celis, O. Reingold, G. Segen, and U. Wieder
\newblock Balls and Bins: Smaller Hash Families and Faster Evaluation
\newblock Foundations of Computer Science (FOCS), 2011 IEEE 52nd Annual Symposium, pages 599 -- 608, 2011

\bibitem{azar}
Y. Azar, A. Broder, A. Karlin, and E. Upfal
\newblock Balanced allocations.
\newblock SIAM Journal on Computing, 29(1). pages 180--200, 1999.

\bibitem{vocking}
B. V\"{o}cking
\newblock How asymmetry helps load balancing.
\newblock In Proceedings of the Fortieth Annual Symposium on Foundations of Computer Science. pages 131--140, 1999.
\end{thebibliography}
\end{document}